\documentclass[twoside]{aiml}

\usepackage{aimlmacro}
\usepackage{amsmath}
\usepackage{stmaryrd}
\usepackage{multicol}
\usepackage{nicefrac}
\usepackage{comment}
\usepackage{xspace}
\usepackage{xcolor}

\newcommand{\Prop}{\mathbb P}

\newcommand{\Sigprime}{\widehat\Sigma}

\newcommand{\Max}{{\rm max}}

\newcommand{\Sub}{{\rm sub}}

\def\lb{\left\llbracket}
\def\rb{\right\rrbracket}

\newcommand{\val}[1]{\lb #1 \rb}

\newcommand{\ps}{\Diamond}

\newcommand{\cut}{\mathcal{M}^{\rm cut}_{\Sigprime}}

\usepackage{amssymb,graphicx}

\newcommand{\dlap}[1]{\makebox[0pt]{\hss#1\hss}}
\newcommand{\cd}{%
  \sbox0{$\ps$}%
  \usebox0\kern-.5\wd0\dlap{\raisebox{.1ex}{\scalebox{.7}[1]{$\cdot$}}}\kern.5\wd0%
}

\newcommand{\cb}{%
  \sbox0{$\Box$}%
  \usebox0\kern-.5\wd0\dlap{\raisebox{.2ex}{\scalebox{.7}[1]{$\cdot$}}}\kern.5\wd0%
}

\newcommand{\face}{\varsigma}
\newcommand{\intor }[1]{#1^\circ}
\newcommand{\faceoff }[1]{#1}
\newcommand{\peq}{\preccurlyeq}
\newcommand{\seq}{\succcurlyeq}
\newcommand{\reach}[2]{\gamma(#1,#2)}
\newcommand{\nc}{\Box}
\def\<{\left <}

\def\>{\right >}

\DeclareSymbolFont{AMSb}{U}{msb}{m}{n}
\DeclareMathSymbol{\N}{\mathbin}{AMSb}{"4E}
\DeclareMathSymbol{\Z}{\mathbin}{AMSb}{"5A}
\DeclareMathSymbol{\R}{\mathbin}{AMSb}{"52}
\DeclareMathSymbol{\Q}{\mathbin}{AMSb}{"51}
\DeclareMathSymbol{\I}{\mathbin}{AMSb}{"49}
\DeclareMathSymbol{\C}{\mathbin}{AMSb}{"43}

\newtheorem{exam}[thm]{Example}

\def\lastname{Bezhanishvili, Bussi, Ciancia, Duque and Gabelaia}

\begin{document}

\begin{frontmatter}
  \title{Logics of polyhedral reachability}
  \author{Nick Bezhanishvili}
  \address{University of Amsterdam,  Amsterdam, The Netherlands}
  \author{Laura Bussi}
  \author{Vincenzo Ciancia}
  \address{National Research Council, Pisa, Italy}
  \author{David Fernández-Duque}
  \address{University of Barcelona, Barcelona, Spain}
  \author{David Gabelaia}
  \address{TSU Razmadze Mathematical Institute, Tbilisi, Georgia}

\begin{abstract}
 Polyhedral semantics is a recently introduced branch of spatial modal logic, in which modal formulas are interpreted as piecewise linear subsets of an Euclidean space. Polyhedral semantics for the basic modal language has already been well investigated. However, for many practical applications of polyhedral semantics, it is advantageous to enrich the basic modal language with a reachability modality. Recently, a language with an Until-like spatial modality has been introduced, with demonstrated applicability to the analysis of 3D meshes via model checking. In this paper, we exhibit an axiom system for this logic, and show that it is complete with respect to polyhedral semantics. The proof consists of two major steps: First, we show that this logic, which is built over Grzegorczyk's system $\mathsf{Grz}$, has the finite model property. Subsequently, we show that every formula satisfied in a finite poset is also satisfied in a polyhedral model, thereby establishing polyhedral completeness.
  
  \end{abstract}

\end{frontmatter}

\section{Introduction}

\parindent 4mm 

Spatial modal logic is a well-established subdiscipline of modal logic, see e.g., \cite{HBSL}.  Its primary focus lies in reasoning about spatial entities and their interrelations. In the topological semantics of modal logic, the modal operators $\Diamond$ and $\Box$ are interpreted as the topological operators of closure and interior, respectively. The classic result of McKinsey and Tarski  states that the modal logic of all topological spaces  is $\mathsf{S4}$. Moreover,  $\mathsf{S4}$ is the logic of any dense-in-itself metric space \cite{Tarski44}. For modern proofs of this and related topological completeness results, we refer to \cite{vBB07}.

\medskip

Recently, a variant of topological semantics was introduced for polyhedra. Polyhedra can be seen as piecewise linear subsets of an $n$-dimensional Euclidean space. For each polyhedron $P$, one interprets formulas into a Boolean algebra of its subpolyhedra. It is easy to see that the closure of a polyhedron is again a polyhedron, thereby providing a \emph{polyhedral semantics for modal logic}. From the standpoint of domain modeling/language expressiveness, polyhedral semantics is easily seen to encompass 3D meshes (which is the natural application domain). As a matter of expressiveness, digital images can be considered as polyhedral models, especially in contexts such as medical imaging, where pixels and voxels (`volumetric
picture elements') have a dimensionality and are considered hyperrectangles.

\medskip

Polyhedral semantics has been introduced and studied in a sequence of papers \cite{BMMP18,gabelaia2018,DBGM22,DBGM23} encompassing both intuitionistic and modal frameworks. These realms are interconnected through the G\"odel translation and the theory of modal companions \cite{CZ97}. In this paper, we focus specifically on modal logics. It follows from \cite{BMMP18} and \cite{DBGM22} that the modal logic of all polyhedra is Grzegorczyk's modal logic $\mathsf{Grz}$. This is the modal logic of finite posets \cite{CZ97}.
In \cite{DBGM22}, a general criterion for a modal or intermediate logic to be complete for polyhedra, the so-called ``nerve criterion'', has been established, which enabled showing that many well-known modal logics such as $\mathsf{Grz.2}$ and $\mathsf{Grz.3}$ are not polyhedrally complete. On the other hand, Scott's logic and logics axiomatized by Jankov formulas of particular, star-like  trees are polyhedrally complete. The logic of convex polyhedra was studied in \cite{DBGM23}, and the full characterization of polyhedral logics of flat polygons was announced in \cite{gabelaiatacl}.

\medskip

However, for many applications, it is important to enrich the modal language with an Until-like \emph{spatial reachability modality}. A prominent example is the research line on \emph{spatial model checking} (see \cite{CBLM21} and the references therein for a lightweight introduction), where the basic modal language is enhanced with reachability \cite{CLLM16}, and interpreted on finitely representable spaces, such as images, graphs, or polyhedra. Further examples of such applications are shown in~\cite{CGLMV23,CGGLMV23,BCGJLMV24}, where authors define bisimilarity modulo reachability. So far, the methodology has been applied in a variety of application domains, among which we mention medical imaging (see e.g. \cite{BCLM19}) and analysis of video streams \cite{BCGLM22}. 

\medskip

Formally, the reachability modality, which we denote by $\gamma$, is interpreted as follows: $\gamma(\varphi, \psi)$ is true at a point $x$ if \emph{there is a path starting at point $x$ and ending at some point $y$ satisfying $\psi$, and every intermediate point along the path satisfies $\varphi$}. Polyhedral semantics with this modality was investigated in \cite{BCGGLM22}. We will give one illustrative example of the use of this modality. In Figure~\ref{fig:maze}, there is a maze represented as triangulated polyhedron (a finite union of points, segments, triangles and tetrahedra). We can think of the red cube as the actual state, green cubes as safe exits, darker cubes as unsafe passage rooms, and white cubes as safe passage rooms. All these rooms are being connected by corridors. Then, the formula $\texttt{red}\land \gamma(\texttt{red}\lor\texttt{corridor}\lor\texttt{white},\texttt{green})$  is true at the red state if and only if there is a safe exit out of the current state of the maze. For more examples and details, we refer to \cite{BCGGLM22}.

\begin{figure}
	\centering{
	\includegraphics[width=.7\textwidth]{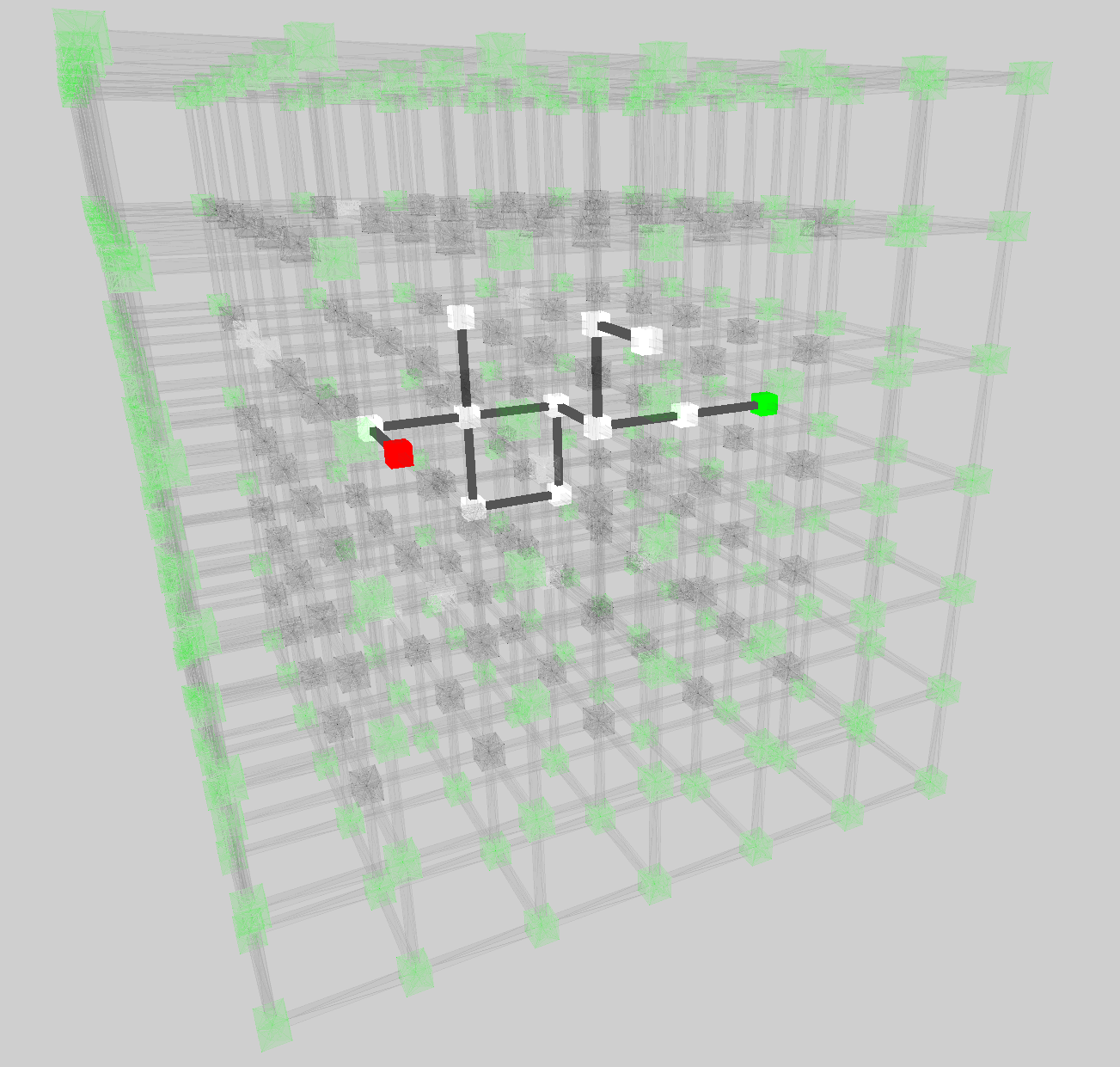}
	}
	\caption{\label{fig:maze}Example of a polyhedral model depicting a maze, a starting area (in red), exit areas (in green), safe places (white), and a `path' witnessing that from the starting area, an exit can be safely reached. All the points in the red area make the formula $\texttt{red}\land \gamma(\texttt{red}\lor\texttt{corridor}\lor\texttt{white},\texttt{green})$ true.}
\end{figure}


In this paper, we introduce the logic $\mathsf{PLR}$, the \emph{polyhedral logic of reachability}. Although in polyhedral semantics, as well as in the posets, the modality $\Diamond$ can be expressed via $\gamma$ with the formula $\gamma(\varphi, \top)$, we still keep $\Diamond$ in the language for convenience. The reachability-free modal fragment of $\mathsf{PLR}$ is the modal logic $\mathsf{Grz}$. We provide axioms for reachability. Our main result states that $\mathsf{PLR}$ is sound and complete with respect to the class of all polyhedra. One of the key insights of polyhedral logics is the fact that all the logical information concerning a polyhedron is encoded in the face posets of its triangulations \cite{DBGM22,BCGGLM22}. 

\medskip

We  prove our main result in stages. To start with, we show that the formula of our language with reachability is satisfiable in a polyhedral model iff it is satisfiable on a finite poset. This is shown using the results of \cite{BMMP18} and \cite{DBGM22} to construct from each finite poset $F$ a nerve poset $N(F)$ consisting of all non-empty chains of $F$ ordered by inclusion. Utilizing the results from \cite{BCGGLM22}, we demonstrate that if $F$ refutes a formula in our language, then so does $N(F)$. Again relying on the results from \cite{BMMP18} and \cite{DBGM22}, we construct a triangulated polyhedron $P$ whose face poset is $N(F)$. We then show that a polyhedral model and its face poset model satisfy precisely the same formulas. This implies that the polyhedron $P$ also refutes the formula refuted on $F$. The converse direction -- that a formula refutable on a polyhedron is refutable on a finite poset -- also follows from the correspondence between polyhedral models and their face poset models, so the reachability logic of polyhedra is shown to be the same as the reachability logic of finite posets. Next we syntactically define a logic $\mathsf{PLR}$. To show that it is sound and complete with respect to finite posets we first prove completeness for an intermediary logic  $\mathsf{ALR}$ (which is based on $\mathsf{S4}$) employing a variant of filtration, and then extend this to $\mathsf{PLR}$ (based on $\mathsf{Grz}$), by the method of cutting clusters in a filtrated model. Combining the above results we obtain that $\mathsf{PLR}$ is sound and complete with respect to polyhedral semantics.

\medskip

The paper is organized as follows: in Section~\ref{sec:back} we provide the reader with background information about topological spaces and polyhedra. Section~\ref{s:semantics} introduces the reachability language and establishes its topological, Kripke and polyhedral semantics. Section~\ref{s:nerve} recalls the key construction of the nerve of a poset, which links polyhedral and Kripke semantics. In Section~\ref{s:logics} we provide an axiomatic definition of the two target logics, $\sf ALR$ and $\mathsf{PLR}$. In Section~\ref{sec:s4} we present our first result, the finite model property of the $\mathsf{S4}$-reachability logic $\sf ALR$. In Section~\ref{sec:grz} we extend this to the finite model property with respect to posets of the $\mathsf{Grz}$-reachability logic $\mathsf{PLR}$ and prove the polyhedral completeness of $\mathsf{PLR}$. In Section~\ref{sec:conc} we draw conclusions and highlight some directions for future work.

\section{Alexandroff spaces and polyhedra}\label{sec:back}

In this section we briefly review some notions from topology, particularly polyhedra, setting up the stage for our semantics.
We begin by reviewing Alexandroff spaces, which provide a link between topological semantics and the familiar Kripke semantics for logics above $\sf S4$.

\subsection{Topological and Alexandroff spaces}

We assume the familiarity with the basic concepts of topology such as topological spaces, closed and open sets, closure and interior, 
basis, etc. We refer to \cite{Engelking} for all these notions. 
Recall that each preordered set (i.e., a set with a reflexive and transitive relation) can be viewed as a special topological space, in which an arbitrary intersection of open sets is open.
Such spaces are known as \emph{Alexandroff spaces}. For a preordered set $(X, R)$ and $x\in X$ we let $R[x] = \{y\in X\mid x R y\}$. Call $U \subseteq X$ an $R$-upset if $R(U) = U$ where $R(U)=\{y\in X\mid \exists x\in U(x R y)\}$. An $R$-downset is defined dually, and for a partially ordered set we
simply say an upset or downset.
The collection $\tau_R$ of all $R$-upsets of $(X, R)$ is an Alexandroff
topology on $X$ such that the closure of a set $U$ is the set $R^{-1} (U) = \{y\in X: R[y]\cap U\neq \varnothing\}$,
and $\{R[x] : x \in X\}$ is a basis for $\tau_R$. 
Conversely, given a topological space $(X, \tau)$ one can define a \emph{specialization order} $R_\tau$ on $X$ by saying that 
$x R_\tau y$ iff every open set containing $x$ also contains $y$. It is well known that if $(X, \tau)$ is an Alexandroff space, 
then $\tau = \tau_{R_\tau}$, and for every preordered set $(X, R)$ we have $R = R_{\tau_R}$. Because of this we will not distinguish 
Alexandroff spaces and preordered sets. Note that any finite topological space is clearly Alexandroff.

\subsection{Simplicial complexes and polyhedra}\label{subsecPoly}

We recall here some definitions about polyhedra. Most of the
definitions are drawn from~\cite{DBGM22,BCGGLM22}.

\begin{definition}
A \emph{d-simplex} $\sigma$ is the convex hull of a finite set $V = \{v_0,\ldots,v_d\} \subseteq \mathbb{R}^m$ of $d + 1$ affinely independent points.
\end{definition}

We recall that $ v_0,\ldots,v_d $ are affinely independent if $v_1-v_0,\ldots,v_d-v_0$ are linearly independent. The number $d$ is said to be the \emph{dimension} of $\sigma$, while the points $v_0,\ldots,v_d$ are said to be its \emph{vertices}.
Note that simplices are bounded, convex and compact subspaces of $\mathbb{R}^n$.

A \emph{face} of $\sigma$ is the convex hull $\tau$ of
a set $T \subseteq \{ v_0,\ldots,v_d \}$, with $T \neq 
\emptyset$: it is straightforward to see that $\tau$
is also a simplex, and we thus have a partial order on simplices given by $\tau \peq \sigma$ if $\tau$ is a face of $\sigma$.

\begin{definition}
The \emph{relative interior} of a simplex $\sigma$ is the set
\[\intor {\sigma} = \{ \Sigma^d_{i=0} \lambda_i v_i \mid \forall i . \ \lambda_i \in (0,1] \ and \ \Sigma^d_{i=0} \lambda_i = 1 \}.\]\end{definition} 

The relative interior of a simplex coincides with its topological interior in the subspace defined by the affine span of the simplex. Any simplex can be partitioned into the relative interiors of its faces. We write $\intor {\sigma} \mathrel{\intor {\peq}} \intor {\tau}$ whenever $\sigma  \peq \tau$, and it is easy to see that this is also a partial order, isomorphic to $\peq$. 

More complex spaces, aptly known as \emph{simplicial complexes,} are obtained by taking finite unions of simplices.

\begin{definition}
A \emph{simplicial complex} $K$ is a finite set of simplices of $\mathbb{R}^n$ such that:
\begin{itemize}
	\item[1.] If $\sigma \in K$ and $\tau \peq \sigma$, then $\tau \in K$;
	\item[2.] If $\sigma, \tau \in K$ with $\sigma \cap \tau\neq\varnothing$, then $\sigma \cap \tau \peq \sigma$.
\end{itemize}
\end{definition}

Simplicial complexes inherit the relations $\intor {\peq}$ and $\peq$.
The dimension of a simplicial complex $K$ is the maximum of the dimensions of its simplices, while the face relation $\peq$ is given by the union of the face relations of the simplices composing $K$. Given a simplicial complex $K$, the \emph{polyhedron} of $K$ (denoted by $|K|$) is the set-theoretic union of its simplices. Most importantly, given a polyhedron $| K |$, each of its points belongs to one and only one of the relative interiors of its simplices, so that the elements of $\intor {K}$ induce a partition of $| K |$:

\begin{lemma}
Each point of $| K |$ belongs to the relative interior of exactly one  simplex in $ K $. That is, $\intor {K} = \{ \intor {\sigma} \mid \sigma \in K \}$ is a partition of $|K|$.
\end{lemma}

$\intor {K}$ is called a \emph{simplicial partition} of $|K|$, and elements of $\intor {K}$ are called \emph{cells}.
Note that the simplicial complex decomposition of a polyhedron is not unique.

We next recall the notion of a \emph{topological path}.

\begin{definition}
A \emph{topological path} in a topological space $X$ is a continuous function $\pi :[0,1] \rightarrow X$, where $[0,1]$ is equipped with the subspace topology of $\mathbb{R}$.
We say that $\pi$ is a path {\em from $a$ to $b$} if $\pi(0)=a$ and $\pi(1)=b$.

If $P=|K|$ for a given simplicial complex $K$ and $\sigma\in K$, we say that $\pi$ {\em traverses $  \sigma$} if there is $x\in (0,1)$ such that $\pi(x)\in \intor \sigma$.
\end{definition}

Topological paths have long been at the core of methods in fields such as algebraic topology and complex analysis.
As we will see, they also give rise to an `until-like' operator in modal logics of space -- the main feature of our logics.

\section{Language and semantics}\label{s:semantics}

In this section, we introduce our logic and its semantics.
The semantics will have three variants, based on Kripke frames, polyhedra, and topological spaces, with the first two being special cases of the third.

\subsection{Language with reachability modality}

Our logics are based on the language generated by the following grammar:

\[
\varphi \equiv \Prop \mid \neg\varphi \mid \varphi\land\psi \mid \Box\varphi \mid \gamma(\varphi,\psi) 
\]

Here, $\Prop$ denotes a set of propositional variables typically assumed to be countably infinite.
We let ${\mathcal L_\gamma}$ denote the set of all formulas of this language.
In this paper, `formula' refers exclusively to elements of $\mathcal L_\gamma$.

\subsection{Topological and polyhedral semantics}

Next we introduce topological semantics for ${\mathcal L_\gamma}$.
As special cases, we obtain polyhedral and Kripke semantics.

\begin{definition}\label{defSem}
A {\em topological model} is a triple $\mathcal X=(X,\tau,\val\cdot)$, where $(X,\tau)$ is a topological space and $\val\cdot\colon\Prop\to 2^{X}$ is a valuation map.
The interior operator for $\tau$ is denoted $\mathcal I$.

Given a topological model $\mathcal X=(X,\tau,\val\cdot)$ we extend $\val\cdot$ to all formulas $\varphi\in{\mathcal L_\gamma}$ via the satisfaction relation $\models$ by stipulating $\val\varphi=\{x\in X\mid \mathcal{X},x \models\varphi\}$ and using the following recursive definition:

\begin{itemize}
	\item $\mathcal{X},x \models p \iff x \in \val p$
	\item $\mathcal{X},x \models \lnot \varphi \iff \mathcal{X},x \not \models \varphi$
	\item $\mathcal{X},x \models \varphi_1 \land \varphi_2 \iff \mathcal{X},x 
	\models \varphi_1$ and $\mathcal{X},x \models \varphi_2$
	\item $\mathcal{X},x \models \nc \varphi \iff x \in \mathcal{I} (\llbracket 
	\varphi \rrbracket)$
	\item $\mathcal{X},x \models \reach \varphi \psi \iff$ there exists a topological path $\pi$ 
	such that $\pi(0) = x$, $\pi(1) \in \llbracket \psi \rrbracket$ and $\pi[(0,1)] 
	\subseteq \llbracket \varphi \rrbracket$
\end{itemize} 
\end{definition}

On occasion we may write $x\models\varphi$ instead of $\mathcal X,x\models\varphi$ when $\mathcal X$ is clear from context.
With this, we may define polyhedral models as a special case of topological models.

\begin{definition}
A {\em polyhedral model} is a pair $\mathcal X=(K,\val\cdot)$, where $K$ is a simplicial complex and $\val\cdot\colon \Prop\to 2^{ K }$. We identify $\mathcal X$ with the topological model $\mathcal X_{\rm top} =  (|K|,
\tau_K,\val\cdot_{\rm top})$, where $\tau_K$ is the subspace topology on $|K|$ inherited from $\mathbb R^n$ and $\val p_{\rm top} = \bigcup_{\sigma \in \val p}\intor  \sigma$ for every $p\in\Prop$. 
\end{definition}

It is easily observed that if $\mathcal{X} = ( K, \val\cdot)$ and $\mathcal{X'} = ( K', \val\cdot)$ are two models with $|K|=|K'|$ (and sharing their propositional valuation), then for each $x \in |K|$ and $\varphi\in{\mathcal L_\gamma}$ we have that $\mathcal{X}, x \models \varphi \iff \mathcal{X'}, x \models \varphi$.

\begin{definition}
An {\em Alexandroff model} is a triple $\mathcal M=(W,\peq,\val\cdot)$, where $W$ is any set, $\peq$ is a preorder on $W$, and $\val\cdot\colon\Prop\to 2^W$.

An Alexandroff model $\mathcal M=(W,\peq,\val\cdot)$ is a {\em finite poset model} if $\peq$ is a partial order and $W$ is finite.

The valuation $\val\varphi$ is extended to arbitrary formulas $\varphi$ by identifying $\mathcal M$ with the topological model obtained by equipping $W$ with the upset topology.
\end{definition}

As usual, we write $w\prec v$ if $w\peq v$ and $v\not\peq w$.
Our Alexandroff models are simply $\sf S4$ models, but we use this terminology to stress that $\sf S4$ frames coincide with Alexandroff spaces equipped with their specialization preorder, allowing us to apply Definition~\ref{defSem} to Alexandroff models.
However, in this setting it will also be convenient to characterize the semantics of $\reach\cdot\cdot$ without an explicit reference to continuous paths.

\begin{definition}
Given an Alexandroff model $\mathcal M=(W,\peq,\val\cdot)$ and a formula $\varphi$, we define the {\em reachability relation} $R^\varphi$ to be the least relation so that $w\mathrel R^\varphi v$ if there is $u \in \val \varphi$ such that either

\begin{enumerate}
	\item $w\peq u \seq v$, or
	\item $w \mathrel R^\varphi u$ and $u \mathrel R^\varphi v$.
\end{enumerate}
\end{definition}

This relation can equivalently be presented in terms of {\em up-down paths.}

\begin{definition} 
Let $\mathcal{M } =(W,\peq,V) $ be an Alexandroff model.
A sequence $(w_0,\ldots,w_k) \subseteq W$ is said to be an 
{\em up-down path} if $k = 2j$ for some $j>0$, $w_0  \peq  w_1$,  $ w_{k-1}  \seq  w_k$, and whenever $0< i<j$, we have that $w_{2i-1} \succ w_{2i}\prec w_{2i+1}$.
\end{definition}

Thus, an up-down path is a path $w_0\peq w_1\succ w_2\prec w_3\succ\ldots \prec w_{k-1}\seq w_k$.

\begin{lemma}\label{lemmUDR}
Given a finite Alexandroff model $\mathcal{M } =(W,\peq,\val\cdot) $, $w,v\in W$, and a formula $\varphi$, the following are equivalent: 
\begin{enumerate}

\item $w\mathrel R^\varphi v$.

\item There is an up-down path $(w_0,\ldots,w_k) \subseteq W$ such that $w_0=w$, $w_k=v$, and for all $i\in(0,k)$, $w_i\in \val\varphi$.

\item There is a topological path $\pi\colon [0,1]\to W$ such that $\pi(0) =w$, $\pi(1) =v$, and for all $t\in (0,1) $, $\pi(t)\in \val\varphi$.

\end{enumerate}
\end{lemma}

\begin{proof}
We prove that the second item implies the first; the other direction is similar.
Suppose that $(w_0,\ldots,w_k) \subseteq W$ is an up-down path such that for all $i\in(0,k)$, $w_i\in \val\varphi$.
We prove by induction on $k$ that $w_0\mathrel R^\varphi w_k $.
 
We know that $w_1\seq w_2$, so that $w_0\peq w_1\seq w_2$ witnesses $w_0\mathrel R^\varphi w_2$ by the `base case' of the definition of $R^\varphi$.
If $k=2$ we are done, otherwise $k>2$ and the induction hypothesis yields $w_2\mathrel R^\varphi w_k$.
Since in this case $w_2\in\val\varphi$, this witnesses the inductive clause for $R^\varphi$, and once again $w_0\mathrel R^\varphi w_k $.

To show that (ii) implies (iii) suppose that $(w_0,\ldots,w_{2k})$ is an up-down path. To build a corresponding topological path, consider points $\frac{i}{k}\in[0,1]$ with $0\leq i\leq k$ and let $\pi(x)=w_{2i}$ if $x=\frac{i}{k}$ and $\pi(x)=w_{2i+1}$ if $x\in\left(\frac{i}{k},\frac{i+1}{k}\right)$. Then $\pi$ is continuous and satisfies the requirements of (iii).

For the other direction suppose $\pi:[0,1]\to W$ is a topological path with $\pi(0)=w$, $\pi(1)=v$ and $\pi[(0,1)]\subseteq\val\varphi$. Since $(0,1)$ is a connected subspace of $[0,1]$ and $\pi$ is continuous, $\pi[(0,1)]$ is a connected subspace of the Alexandroff space $W$. Moreover, it follows from the continuity of $\pi$ that $w=\pi(0)$ and $v=\pi(1)$ are in the closure of $\pi[(0,1)]$. Hence there exist  $w_1,v_1\in \pi[(0,1)]$ with $w\peq w_1$ and $v\peq v_1$.
It is not hard to check that if $A$ is a connected subspace of an Alexandroff space, then any two elements of $A$ are connected via a zigzag path \cite[Lemma 3.4]{Bezh-Gab-2011}.
In particular, the connectedness of $\pi[(0,1)]\subseteq\val\varphi$ ensures now the existence of an up-down path as required in (ii).
\end{proof}

Given this context, up-down paths are closely related to topological paths.
To further utilize this connection, first we define a simplicial model which is a Kripke companion to a given polyhedral model.

\begin{definition} 
For a polyhedral model $\mathcal{X} = (K,\val\cdot)$, the {\em Kripke companion of $\mathcal{X}$} is $\mathcal{M(X)} = (\faceoff {K}, \faceoff {\peq}, \faceoff {\val\cdot})$, where simplices of $K$ are regarded as points, $\faceoff {\peq} $ is the face relation of Section~\ref{subsecPoly}, and ${\val{p}}$ is unchanged.

We define $\face  \colon |K| \to \faceoff K $ by letting $\face (x) $ be the unique $\faceoff \sigma\in \faceoff K$ such that $x\in\intor \sigma$.
\end{definition}

\begin{exam}
Let us illustrate with a simple example of a polyhedron consisting of a single triangle with vertices $a,b,c$, including its interior.
Then, $|K|$ consists of the set of points in the closed triangle, while $K=\{\{a\},\{b\},\{c\}, \overline{bc},\overline{ac},\overline{ab},|K|\}$, where $\overline{x y}$ denotes the closed line segment from $x$ to $y$.
For $x\in |K|$, $\face(x)$ is the smallest face to which $x$ belongs.
Thus for example, $\face(a) = \{a\}$ since vertices are `singleton faces', $\face(x) = \overline{ab}$ if $x$ is {\em strictly } between $a$ and $b$, and $\face (x) = |K|$ only when $x$ is in the topological interior of $|K|$, i.e.~if it does not lie on any of the edges of our triangle.
\end{exam}

\begin{lemma}\label{lemmUDS}
Let $\mathcal{X} = (K,\val\cdot)$ be a polyhedral model.
Then, the map $\face   \colon |K| \to \faceoff K $ is continuous with respect to the Alexandroff topology on $K$.

Moreover, if $x,y\in |K|$ and $N\subseteq K$, the following are equivalent:
\begin{enumerate}

\item\label{itUDSOne} There is a topological path from $x$ to $y$ traversing all the elements of $N$.

\item\label{itUDSTwo} There is an up-down path $(\sigma_0,\ldots,\sigma_n)$ from $\face (x)$ to $\face (y)$ such that $N = \{\sigma_1,\ldots,\sigma_n\}$.

\end{enumerate}
\end{lemma}

\begin{proof}
First we check that $\face $ is continuous.
For this it suffices to show that if $C\subseteq K$ is closed, then $\face^{-1}[C] =\{x\in |K|:x\in \bigcup \intor C\}$ is closed.
We have that $C$ is closed iff it is downward-closed under $\peq$, i.e.~if $\tau\peq \sigma\in C$ then $\tau\in C$.
From this it follows that if $\sigma\in C$ then $\sigma\subseteq \bigcup \intor C$, as if $x\in \sigma$ it follows that $x\in \intor \tau$ for some $\tau\peq \sigma$ and, as observed, $\tau\in C$.
Since $C$ is closed, we thus have that $\face ^{-1}[C] = \bigcup \intor C = \bigcup   C$, and since the latter is closed, so is the former.

From this it is immediate that \eqref{itUDSOne} implies \eqref{itUDSTwo}, since if $\pi\colon [0,1] \to |K|$ is a topological path then $\face\circ\pi\colon [0,1] \to K$ is also a topological path (being continuous) and clearly the two traverse the same faces.
For the converse, it suffices to consider the case where $N$ has three (possibly repeating) elements $\sigma_0\peq \sigma_1\seq \sigma_2$ and to find a path from $x\in \intor \sigma_0 $ to $y\in \intor \sigma_2 $ traversing only $ \sigma_1$, as such paths can then be strung together inductively.
Choose $z\in \intor \sigma_1$; then, it can easily be checked that the piecewise linear path $\pi$ with $\pi(0)=x$, $\pi(\nicefrac 12) = z$ and $\pi(1) = y$ has the required properties.
\end{proof}

Combining lemmas~\ref{lemmUDR} and \ref{lemmUDS}, we see that the semantics of $\reach\cdot\cdot$ for a polyhedral model coincides with that of its Kripke companion, in the following sense.

\begin{lemma}\label{lemmPoltoK}
Let $\mathcal{X} = (K,\val\cdot)$ be a polyhedral model and $x \in |K|$.
Then, for every formula $\varphi$, we have that $\mathcal{X}, x \models \varphi \iff \mathcal{M(X)}, \face(x) \models \varphi$.
\end{lemma}

It follows readily from Lemma~\ref{lemmPoltoK} that any formula satisfiable on a polyhedral model is also satisfiable on a finite poset model.
The converse is also true, but the transformation from Kripke models to polyhedral ones requires an extra step, as we see in the next section.

\section{The nerve of a Kripke model}\label{s:nerve}

Not every finite poset model $\mathcal M$ is of the form $\mathcal M(\mathcal X)$ for some polyhedral model $\mathcal X$, but it is possible to transform $\mathcal M$ into a new model which {\em does} have this property.
This new model is the {\em nerve} of $\mathcal M$.

\begin{definition}
Given a poset $W$, its \emph{nerve}, $N(W)$, is the collection of finite non-empty chains in $W$ ordered by set-theoretic inclusion. We define a function $\Max:N(W)\to W$ sending each element of $N(W)$ to its maximal element and use this map to `pull back' the valuation from $W$ to $N(W)$; to be precise, we set $\val{p}^N = \{ c \mid \Max(c) \in \val p \}$.
The {\em nerve model} of $\mathcal M$ is then the model $N(\mathcal M) = (N(W),\subseteq,\val\cdot^N)$.
\end{definition}

The map $\Max:N(W)\to W$ mapping a chain to its maximal element is a p-morphism \cite{DBGM22}.
We can always find a polyhedral model $\mathcal X$ such that $N(\mathcal M) = \mathcal M(\mathcal X)$, as shown, e.g.~in \cite{DBGM22}.

\begin{theorem}\label{theoNerve}
If $\mathcal M$ is any finite poset model then there exists a polyhedral model $\mathcal X$ such that $N(\mathcal M) = \mathcal M(\mathcal X)$.
\end{theorem}

The function   $\Max$ will preserve the truth of formulas in our language.
To make this precise, we must introduce the notion of up-down morphism.

\begin{definition} 
Let $\mathcal M =(W,\peq,\val\cdot)$, $\mathcal M'=(W',\peq',\val\cdot')$ be finite poset models. A function $f: W \rightarrow W'$ is an {\em up-down morphism} if:
\begin{itemize}
	\item (atom) if $w \in W$, $w \in \val p$, then $f(w) \in \val p'$;

	\item (forth) if $w \peq u \seq v$, then $f(w) \peq f(u) \seq f(v)$, and

	\item (back) if $f(w) \peq' u' \seq' v'$, then there is an up-down path $(v_0,\ldots,v_k)\subseteq W$ such that $v_0=w $, $f(v_k)=v'$, and for all $i\in(0,k)$, $f(v_i) = u'$.

\end{itemize}
\end{definition}

It is easy to see that an up-down morphism is also a p-morphism. Indeed, if $w \in W$ and $w \peq v$, then we have an up-down path $w \peq v \seq v$, and this implies $f(w) \peq' f(v) \seq' f(v)$, so in particular $f(w) \peq' f(v)$.
On the other hand, if $f(w) \peq' v'$, then we consider the up-down path $f(w) \peq' v' \seq' v'$ and we get that there is an up-down path $(v_0,\ldots,v_n)\subseteq W$ such that $v_0=w$ and $f(v_1) = v'$, which yields the back condition for $\peq$  since $v_0\peq v_1$.

\begin{lemma}\label{lemmPmorphPreserve}
Let $\mathcal M =(W,\peq,\val\cdot)$, $\mathcal M'=(W',\peq',\val\cdot')$ be finite poset models and suppose that $f: W \rightarrow W'$ is an up-down morphism.
Then, for every formula $\varphi\in{\mathcal L_\gamma}$ we have $\val\varphi=f^{-1}[\val\varphi']$.\end{lemma}

\begin{proof}
The proof proceeds by a standard structural induction on $\varphi$.
The case for $\Box \varphi $ follows from up-down morphisms being also p-morphisms.
The case for $\reach\varphi\psi$ requires observing that if $(w_0,\ldots,w_{2n})$ is an up-down path in $W$ witnessing that $\reach\varphi\psi$ holds on $w = w_0$, then the `forth' clause can be applied $n $ times to obtain a path $(w'_0,\ldots,w'_{2n'})$ witnessing $\reach\varphi\psi$ on $f(w_0)$.
If instead $\reach\varphi\psi$ holds on $f(w)$, a symmetric argument shows that  $\reach\varphi\psi$ holds on $ w $.
\end{proof}

To see that the nerve of a Kripke model satisfies the formulas satisfied in the original model, it suffices to see that $\Max$ is an up-down morphism.

\begin{lemma}\label{lemmNerveMorph}
If $\mathcal M=(W,\peq,\val\cdot)$ is a finite poset model then the map $\Max\colon N(W)\to W$ is an up-down morphism.
\end{lemma}

\begin{proof}
The atomic clause is taken care of by the definition of $\val\cdot^N$.
For the (forth) condition, if $c, c_1, c_2 \in N(W)$ with $c \subseteq c_1 \supseteq c_2$, it readily follows that $\Max(c) \peq \Max(c_1) \seq \Max(c_2)$.

Let us now take $w, u, v \in W$ such that $w \peq u \seq v$. We note that $N(W)$ contains the chains $\{w\} \subseteq \{w,u\} \supseteq \{u\} \subseteq \{v,u\} \supseteq \{v\}$, and that this sequence forms an up-down path. Moreover, we have $\Max(\{w\}) = w$, $\Max(\{v\}) = v$ and $\Max(\{w,u\}) = \Max(\{u\}) = \Max(\{v,u\}) = u$. Hence the (back) condition is also satisfied.
\end{proof}

It follows that the satisfiability in finite poset models is the same as the satisfiability in polyhedral models.

\begin{theorem}\label{theoKripkePoly}
A formula $\varphi$ is satisfiable on the class of finite poset models if and only if it is satisfiable on the class of polyhedral models.
\end{theorem}

\begin{proof}
That any formula satisfied on a polyhedral model is satisfied on a finite poset model follows from Lemma~\ref{lemmPoltoK}.

For the converse, if $\varphi$ is satisfied on a finite poset model $\mathcal M$, then $\varphi$ is also satisfied on $N(\mathcal M)$ by lemmas~\ref{lemmPmorphPreserve} and \ref{lemmNerveMorph}.
Theorem~\ref{theoNerve} and Lemma~\ref{lemmPoltoK} then imply that $\varphi$ is satisfied on a polyhedral model.
\end{proof}

We conclude that the logic of polyhedral models coincides with the logic of finite poset models.
With this in mind, we may restrict our attention to the latter for the remainder of the text.

\section{Logics of reachability}\label{s:logics}

In this section, we introduce axiomatically two key logics of this paper -- $\sf ALR$ and  $\sf PLR$.
The former will be based on $\sf S4$ and the latter on $\sf Grz$, both familiar modal logics related to transitive frames and topological spaces.  
However, it is   necessary to introduce new axioms and rules for the reachability operator. In particular, we define two rules: the first one is quite intuitive, and it basically states that implication preserves reachability, or that reachability is monotone in both coordinates.

The second rule, or \emph{induction rule}, is a bit more subtle.
Intuitively, it works as a back propagation rule for the $\psi$ formula, where the first premise is needed to propagate through a down step of an up-down path, while the second one is needed to propagate through an up step.

\begin{definition}
Axioms of the Alexandroff reachability logic $\sf ALR$ are given by all the propositional tautologies and Modus Ponens, $\sf S4$ axioms and rules for $\nc$, plus the following:

\begin{description}

	\item[Axiom 1.] $\psi \vee (\varphi\wedge \reach\varphi\psi) \to\nc(\varphi\to\reach\varphi\psi)$
	
	\item[Axiom 2.] $ \ps (\varphi\wedge  \reach\varphi\psi )\to \reach\varphi\psi $
	\medskip
	\item[Rule 1.]
	$\displaystyle\dfrac{\varphi\to\varphi' \ \ \psi\to\psi'}{\reach\varphi\psi \to \reach{\varphi'}		{\psi'}}$
		\medskip

	\item[Rule 2.] 
	$\displaystyle\dfrac{\psi\to \nc  (\varphi\to \psi)  \ \ \ \varphi\wedge \ps(\varphi\wedge\psi) \to   	\psi   }{ \reach \varphi\psi  \to \ps(\varphi\wedge\psi)}.$

\end{description}
	\medskip

The polyhedral reachability logic $\sf PLR$ is obtained by adding the $\sf Grz$ axiom $\Box(\Box(p\to\Box p)\to p)\to \Box p$ to $\sf ALR$.
\end{definition}

We note in passing that the above two rules can also be formulated as axioms using a defined operator $[\pi]\varphi\equiv \neg\gamma(\top,\neg\varphi)$ which semantically acts as a `global' box modality inside the path-connected component of the point of evaluation. Indeed, it is not hard to see that $[\pi]\varphi$ will be true at a point $w$ in an Alexandroff model iff $\varphi$ is true at each $v$ which connects with $w$ by a zigzag path, that is $\varphi$ is universally true inside the connected component containing $w$. With this in mind, Rule 1 can be expressed as:
 $$[\pi](\varphi\to\varphi')\land [\pi](\psi\to\psi')\land \reach\varphi\psi \to \reach{\varphi'}{\psi'}$$

The next proposition shows soundness of our logics.

\begin{proposition}\label{prop:sound}
The logic $\sf ALR$ is sound for the class of Alexandroff spaces and the logic $\sf PLR$ is sound for the class of finite poset models. 
\end{proposition}
\begin{proof}
	Let $\mathcal M=(W,\peq,\val\cdot)$ be an Alexandroff model.

	Axiom 1: $\psi \vee (\varphi\wedge \reach\varphi\psi) \to\nc(\varphi\to\reach\varphi\psi)$
	\smallskip
	
	Take $w \not \models \nc (\varphi \to \reach \varphi \psi)$: then $\exists v \seq w .
	\ v \not \models (\varphi \to \reach \varphi \psi)$, namely $v \models \varphi$, $v \not 
	\models \reach \varphi \psi$. Suppose $w \models \psi \vee (\varphi\wedge \reach\varphi\psi)$. We then have two cases:
	
	\begin{itemize}
		\item $w \models \psi$: we have $v \peq v \seq w$. Hence $v \mathrel R^{\varphi} w$ and 
		$v \models \reach\varphi\psi$, a contradiction.
		\item $w \models (\varphi \land \reach\varphi\psi)$: again we have $v \peq v \seq w$, hence 
		$v \mathrel R^{\varphi} w$. From $w \models \reach\varphi\psi$ we also get $w \mathrel R^{\varphi} u$ for some $u \models \psi$. But then, since $w\models\varphi$, $v \mathrel R^{\varphi} w$ and $w \mathrel R^{\varphi} u$ we obtain $v \mathrel R^{\varphi} u$. Together with $u \models \psi$ this implies $v \models 
		\reach\varphi\psi$, again a contradiction. 
	\end{itemize}

	Axiom 2: $\ps (\varphi \land \reach \varphi \psi) \to \reach \varphi \psi$
	\smallskip
	
	Take $w \models \ps (\varphi \land \reach \varphi \psi)$, $w \not \models \reach \varphi \psi$. Then there is a point $v$ such that $w \peq v$ and $v \models \varphi$, $v \models \reach \varphi \psi$. Hence $w \mathrel R^\varphi v$ and $v \mathrel R^\varphi u$ for some $u \models \psi$. This implies $w \mathrel R^\varphi u$ and thus, $w \models \reach \varphi \psi$ -- a contradiction.
	\medskip
	
	Rule 1:
	\[ \dfrac{\varphi\to\varphi' \ \ \psi\to\psi'}{\reach\varphi\psi \to \reach{\varphi'}		   {\psi'}} \]
	
	Take a model $(W, R, \val . )$ and suppose that $\forall w \in W$, $w \models \varphi 
	\implies w \models \varphi'$, $w \models \psi \implies w \models \psi'$. Take points 
	$u,v$ such that $u \mathrel R^\varphi v$ and $v \models \psi$. Then it easily follows that $u \mathrel 
	R^{\varphi'} v$ and $v \models \psi'$. Hence
	$\reach \varphi \psi \to \reach {\varphi'} {\psi'}$ is valid in the model.
	
	Rule 2 (inductive rule):
	\[ \dfrac{\psi\to \nc  (\varphi\to \psi)  \ \ \ \varphi\wedge \ps(\varphi\wedge\psi) \to   	   \psi   }{ \reach \varphi\psi  \to \ps(\varphi\wedge\psi)}. \]
	
	Take a model $\mathcal{M}$ such that for all $w \in W$ we have $w \models \psi \to  
	\nc (\varphi \to \psi)$ and $w \models \varphi \land \ps(\varphi \land \psi)
	\to \psi$. Suppose $u\models \reach \varphi \psi$. Then there is an up-down path $w_0\peq w_1\succ w_2\prec w_3\succ\ldots \prec w_{2k-1}\seq w_{2k}$ with $w_0=u$,   $w_{2k}\models\psi$ and $w_i\models\varphi$ for all $i$ with $0<i<2k$. We aim to show that $w_1\models\psi$ thereby demonstrating $u\models\ps(\varphi\wedge\psi)$. We will in fact show by a backwards induction that $w_i\models\psi$ for all $i$ with $0<i<2k$. Indeed, given the aforementioned up-down path, we start by considering the last point (where $\psi$ holds) and we build the proof by proceeding backward on the path. To this end, consider first $w_{2k-1}$. From $w_{2k}\models \psi\to \nc  (\varphi\to \psi)$ and $w_{2k}\models\psi$ we get $w_{2k}\models \nc  (\varphi\to \psi)$. Since $w_{2k}\peq w_{2k-1}$ and $w_{2k-1}\models\varphi$, we have $w_{2k-1}\models\psi$. If $k=1$ we are done, otherwise $w_{2k-2}\models\varphi$ and we continue. From $w_{2k-2}\peq w_{2k-1}\models\varphi\wedge\psi$ we obtain $w_{2k-2}\models \ps(\varphi\wedge\psi)$.  Utilizing $w_{2k-2} \models \varphi \land \ps(\varphi \land \psi)\to \psi$ we conclude $w_{2k-2}\models\psi$. Continuing in this fashion, we will end up with $w_1\models\psi$ as desired.
	
	If moreover $\mathcal M$ is a finite poset model, thus based on a finite poset, then it is well known that the $\sf Grz$ axiom is valid on $\mathcal M$.
\end{proof}

\begin{lemma}\label{lem:diamU}
$\reach \varphi \psi \to \ps \varphi$ is a theorem of $\sf ALR$.
\end{lemma}

\begin{proof}
Since $\psi \to \top$ is a tautology, applying Rule 1 gives $\reach \varphi \psi \to \reach \varphi \top$. Taking easily derivable $\top \to \nc (\varphi \to \top)$ and $\varphi \land \ps (\varphi \land \top) \to \top$ as premises of Rule 2, we conclude $\reach \varphi \top \to \ps (\varphi \land \top)$. Modus Ponens now yields  $\ps \varphi$.
\end{proof}

\section{Alexandroff completeness}\label{sec:s4}

In order to prove completeness of the logic, we will use a widely known technique in the realm of modal logics. We will go through the construction of a \emph{canonical model}, which will be filtered modulo a specifically crafted set of formulas. Finiteness of such a set will give us a finite model, which is essential in showing that the axiom system also enjoys the \emph{finite model property}. We will then prove a filtration lemma, to achieve completeness.

\begin{definition}
Let $\Lambda\in\{{\sf ALR},{\sf PLR}\}$.
The {\em canonical model for $\Lambda$}  is the structure $\mathcal M_\mathrm c = (W_\mathrm c,\peq _\mathrm c, \val\cdot_\mathrm c)$, where $W_\mathrm c$ is the set of all {\em $\Lambda$-theories} (maximal $\Lambda$-consistent sets), $T\peq_c S$ if whenever $\nc \varphi\in T$, it follows that $\nc\varphi\in S$, and $\val p_\mathrm c = \{T\in W_\mathrm c: p\in T\}$.
\end{definition}

The results in this section apply to both the canonical model for $\sf ALR$ and for $\sf PLR$, beginning with the following standard lemma.

\begin{lemma}\label{l:diam-witnessing}
For any formula $\varphi$,  $\ps\varphi\in T\in W_\mathrm c$ iff there is $S\seq_c T$ with $\varphi\in S$.
\end{lemma}

\subsection{Filtration}

Say that a set $\Sigma$ of formulas is {\em adequate} if it is closed under subformulas and single negations, and whenever $\reach \varphi\psi\in \Sigma$, then $\nc(\varphi\to \reach \varphi\psi)\in\Sigma$ and $\ps(\varphi\wedge  \reach \varphi\psi) \in\Sigma$. For a set $\Gamma$ of formulas let $\Sub (\Gamma)$ denote the set of all subformulas of formulas in $\Gamma$.

\begin{lemma}\label{l:adequate-finite}
Let $\Gamma$ be a finite set of formulas. Then the smallest adequate set containing $\Gamma$ is also finite.
\end{lemma}

\begin{proof}
Let $\Gamma_1=\Gamma\cup \{ \nc(\varphi\to \reach \varphi\psi), \ps(\varphi\wedge  \reach \varphi\psi) \mid \reach \varphi \psi \in \Sub (\Gamma) \}$. Clearly $\Gamma_1$ is finite. Moreover, $\Gamma$ and $\Gamma_1$ have the same set of subformulas of the form $\reach \varphi\psi$. 

Let now $\Sigma=\Sub (\Gamma_1)\cup \{\neg\varphi\mid\varphi\in \Sub (\Gamma_1), \varphi\neq\neg\psi\}$.  It is straightforward to check that $\Sigma$ is finite and adequate. 
\end{proof}

\begin{definition}
Let $\Sigma $ be adequate.
Define $T\sim_\Sigma S$ if $T\cap\Sigma=S\cap\Sigma$.
We define $\mathcal M_\Sigma=(W_\Sigma,\peq_\Sigma,\val\cdot_\Sigma)$ as follows:
\begin{enumerate}

\item $W_\Sigma = \{ T_\Sigma :T\in W_\mathrm c\}$, where $ T _\Sigma $ is the equivalence class of $T$ under $\sim_\Sigma$.

\item $T_\Sigma \peq _\Sigma S_\Sigma$ if whenever $\nc\varphi\in T\cap\Sigma$, it follows that $\nc\varphi\in S$.

\item $\val p_\Sigma =\{T_\Sigma:p\in T\cap\Sigma\}$.

\end{enumerate}

It easily follows from the definitions that $S\peq_c T$ implies $S_\Sigma \peq_\Sigma T_\Sigma$. 

Given $\varphi\in\Sigma$, we define a $\varphi$-reachability relation on $\mathcal M_\Sigma$ inductively by letting $T_\Sigma \mathrel R^\varphi_\Sigma S_\Sigma$ if there is $U_\Sigma \in W_\Sigma$ such that $\varphi\in U$ and either

\begin{enumerate}
	\item $T_\Sigma\peq_\Sigma U_\Sigma \seq_\Sigma S_\Sigma$, or
	\item $T_\Sigma \mathrel R^\varphi_\Sigma U_\Sigma$ and $U_\Sigma \mathrel R^\varphi_\Sigma S_\Sigma$.
\end{enumerate}
\end{definition}

It easily follows from this definition that $R^\varphi_\Sigma$ is a symmetric relation. 

\begin{lemma}\label{sound}
If $\reach\varphi\psi \in \Sigma$, $T_\Sigma \mathrel R^\varphi_\Sigma S_\Sigma$, and either $\psi \in S$ or $\varphi\wedge \reach\varphi\psi\in S$, then $\reach\varphi\psi \in  T $.
\end{lemma}

\begin{proof}
By induction on the structure of $R^\varphi_\Sigma$ from $T_\Sigma$ to $S_\Sigma$. \\
In the base case, there is $U_\Sigma$ such that $\varphi\in U$, $T_\Sigma \peq_\Sigma 
U_\Sigma$ and $S_\Sigma \peq_\Sigma U_\Sigma$. \\
Since $\reach \varphi \psi \in \Sigma$ and $\Sigma$ is adequate, we have $\nc (\varphi \to \reach \varphi \psi) \in \Sigma$.
Suppose $\psi \in S$ or $\varphi\wedge \reach\varphi\psi\in S$.
By Axiom 1 and Modus Ponens, $\nc (\varphi \to \reach \varphi \psi) \in S$. From $S_\Sigma \peq_\Sigma U_\Sigma$ it follows that $\nc (\varphi \to \reach \varphi \psi) \in U$ and by the reflexivity axiom of $\sf S4$ we get $\varphi \to \reach \varphi \psi\in U$.
Since $\varphi\in U$, we may conclude $\reach \varphi \psi\in U$.
By adequacy of $\Sigma$ we also have $\ps(\varphi \wedge \reach \varphi \psi ) \in  \Sigma$ and since $T_\Sigma \peq_\Sigma U_\Sigma$, $\varphi\wedge\reach \varphi \psi\in U$, we obtain $\ps(\varphi \wedge \reach \varphi \psi ) \in  T$. By Axiom 2, $\reach \varphi \psi \in  T_\Sigma$, as required.

For the inductive case we have that there is $U$ with $\varphi\in U$, $T_\Sigma \mathrel R^\varphi_\Sigma U_\Sigma$, and $U_\Sigma \mathrel R^\varphi_\Sigma S_\Sigma$.
By the induction hypothesis, since $\psi\in S$ and $U_\Sigma \mathrel R^\varphi_\Sigma S_\Sigma$ we obtain $\reach\varphi\psi\in U$. It follows that $\varphi\wedge \reach\varphi\psi\in U$ and since $T_\Sigma \mathrel R^\varphi_\Sigma U_\Sigma$, once again by the induction hypothesis, we have $\reach \varphi\psi\in T$.
\end{proof}

\smallskip

\begin{definition}
Let $\Sigma$ be an adequate and finite set. For $T_\Sigma\in W_\Sigma$ we define $\chi (T_\Sigma) = \bigwedge(T\cap\Sigma)$ (note that this is well defined by the definition of $\sim_\Sigma$).
\end{definition}

\smallskip

\begin{lemma}
Let $\Sigma$ be an adequate and finite set. Suppose $U_\Sigma \in W_\Sigma$ is a world of the filtrated model, $\varphi\in\Sigma$ and let $\chi = \bigvee\{\chi(s):s\in R^\varphi(U_\Sigma )\}$.
Then,
\begin{enumerate}

\item ${\sf ALR} \vdash \ps(\varphi\wedge\chi)\to\chi$

\item ${\sf ALR} \vdash  \varphi\wedge\chi\to \nc(\varphi\to\chi)$.

\end{enumerate}
\end{lemma}

\begin{proof}
	
We will make use of the easily seen observation that for any $S\in W_c$ we have $\chi\in S$ iff $S_\Sigma\in R^\varphi(U_\Sigma)$.
	
To prove the first item, suppose that $\ps(\varphi\wedge\chi)\to\chi$ is not provable, i.e.~$ \neg \chi \wedge \ps(\varphi\wedge\chi)$ is consistent.
Then there is $T\in W_\mathrm c$ with $\ps(\varphi\wedge\chi)\wedge \neg \chi\in T$. By Lemma~\ref{l:diam-witnessing}, there exists $S\seq_\mathrm c T$ such that $\varphi\wedge\chi\in S$.
Since $\chi\in S$, we have that $S_\Sigma \in R^\varphi(U_\Sigma) $. On the other hand, since
$\varphi \in S$ and $T_\Sigma \peq_\Sigma S_\Sigma\seq_\Sigma S_\Sigma$, we also have $T_\Sigma \mathrel R^\varphi_\Sigma S_\Sigma$ and hence $T_\Sigma \in R^\varphi(U_\Sigma)$, implying $\chi 
\in T$, a contradiction.

Suppose now that $\varphi\wedge\chi\to \nc(\varphi\to\chi)$ is not provable, i.e. 
$(\varphi \land \chi) \land \ps \lnot(\varphi \to \chi)$ is consistent. 
Then there is $T \in W_c$ such that 
$(\varphi \land \chi) \land \ps (\varphi \land \lnot \chi) \in T$, and again by Lemma~\ref{l:diam-witnessing} we get $S \seq_c T$ such that $\varphi \land \lnot \chi \in S$. Since $\chi 
\in T$, we have that $T_\Sigma \in R^\varphi (U_\Sigma)$. 
We now note that $\varphi \in S$, and thus $T_\Sigma \mathrel R^\varphi S_\Sigma$, as
$T_\Sigma\peq_\Sigma S_\Sigma \seq_\Sigma S_\Sigma$. Then $U_\Sigma \mathrel R^\varphi S_\Sigma$ also holds, implying $\chi \in S$, again a contradiction. 
\end{proof}

\smallskip

\begin{lemma}\label{l:compl}
If $\reach\varphi\psi\in T\cap \Sigma$, then there exists $S$ with $T_\Sigma \mathrel R^\varphi_\Sigma S_\Sigma$ and $\psi\in S$.
\end{lemma}

\begin{proof}
Toward a contradiction, assume that $\reach\varphi\psi\in T$ and for all $S_\Sigma \in R^\varphi(T_\Sigma)$ we have that $\psi \notin S$.
Let $\chi$ be as in the previous lemma, for $U_\Sigma=T_\Sigma$. Taking contrapositives we obtain:
\begin{enumerate}

\item ${\sf ALR} \vdash \neg\chi\to \nc (\varphi\to\neg \chi)$

\item ${\sf ALR} \vdash \varphi\wedge \ps(\varphi\wedge\neg\chi)  \to \neg\chi$.

\end{enumerate}
By Rule 2, ${\sf ALR} \vdash  \reach \varphi{\neg\chi} \to \ps(\varphi\wedge\neg \chi)$.

Further, since $\psi \not \in S$ for each $S_\Sigma \in R^\varphi(T_\Sigma)$, we get ${\sf ALR} \vdash \chi \to \lnot
\psi$, hence ${\sf ALR} \vdash \psi \to\neg \chi$. By Rule 1,  ${\sf ALR} \vdash  \reach \varphi{\psi} \to \reach \varphi{\neg\chi}$. Together with ${\sf ALR} \vdash  \reach \varphi{\neg\chi} \to \ps(\varphi\wedge\neg \chi)$ this yields ${\sf ALR} \vdash  \reach \varphi{\psi} \to \ps(\varphi\wedge\neg \chi)$.

Since we assumed $\reach\varphi\psi\in T$ at the outset, $\ps(\varphi\wedge\neg \chi)\in T$ follows, and hence there is $S\seq_c T$ such that $\varphi\wedge\neg\chi \in S$.
Since $S_\Sigma\seq_\Sigma T_\Sigma$ and $\varphi\in S$, we obtain $T_\Sigma\mathrel R^\varphi_\Sigma S_\Sigma$, thus by our definition of $\chi$, $ \chi \in S$, contradicting $\neg\chi\in S$.
\end{proof}

With this we are able to prove a key filtration lemma.

\begin{lemma}\label{lem:quotient}
Let $\Sigma$ be a finite adequate set of formulas and let $\varphi\in \Sigma$. Then for each $T\in W_c$ we have $\varphi\in T$ iff $\mathcal M_\Sigma,T_\Sigma\models \varphi$.
\end{lemma}
\begin{proof}
    Standard induction on formulas using lemmas~\ref{l:diam-witnessing},~\ref{sound} and \ref{l:compl}.
\end{proof}

\begin{thm}
$\sf ALR$ is complete for the class of Alexandroff models and has the finite model property.
\end{thm}

\begin{proof}
Suppose ${\sf ALR}\not\vdash\neg\varphi$. Then there is $T\in W_c$ with $\varphi\in T$. By Lemma~\ref{l:adequate-finite} there exists a finite adequate set $\Sigma$ containing $ \mathrm{sub}(\varphi)$. It follows from Lemma~\ref{lem:quotient} that  $\mathcal M_\Sigma, T_\Sigma\models\varphi$. Since $\mathcal M_\Sigma$ is finite, the proof is finished.
\end{proof}

\section{Polyhedral completeness}\label{sec:grz}

Now that we have obtained completeness for Alexandroff spaces, it remains to show that $\sf PLR$ is complete for the class of finite poset models, hence by Theorem~\ref{theoKripkePoly}, also for the class of polyhedral logics.
The proof builds on that for $\sf ALR$, but requires some additional steps.

For a given adequate set $\Sigma$ we define $\Sigma_1 = \Sigma \cup \{ \ps (\lnot \varphi \land \ps \varphi) \mid \ps \varphi \in \Sigma \} \cup \{ \ps (\varphi \land \lnot \psi) \mid \reach \varphi \psi \in \Sigma \}$ and let $\Sigprime$ be the smallest adequate set containing $\Sigma_1$. It is a straightforward consequence of Lemma~\ref{l:adequate-finite} that if $\Sigma$ is finite, $\Sigprime$ is also finite.

Given a finite adequate set $\Sigma$ we now take the filtration of the canonical model for $\sf PLR$ via $\Sigprime$ to obtain a finite model $\mathcal{M}_{\Sigprime}$. We know from Lemma~\ref{lem:quotient} that for $\varphi\in\widehat\Sigma$, $\mathcal M_{\Sigprime}, T_{\Sigprime} \models \varphi \iff \varphi \in T$. The model $\mathcal M_{\Sigprime}$ may however fail to be a poset. The next definition allows for a necessary transformation of $\mathcal M_{\Sigprime}$.

\begin{definition}
Let $\mathcal{M} = (W, \peq , \val\cdot)$ be an Alexandroff model. We define ${\rm cut}(\mathcal{M}) = (W', \peq ', \val\cdot')$ where $W' = W$, $\val p' = \val p$ for all atomic propositions and for all $x, y \in W'$, $x \peq ' y \iff  x =y$ or $x\prec y$\footnote{Recall that $x\prec y$ is a shorthand for $x\peq y$ and $y\not\peq x$.}.
\end{definition}

Note that if $\peq$ is a preorder, then $\peq '$ is a partial order, so ${\rm cut}(\mathcal{M})$ is a finite poset model.

Let $\cut = {\rm cut}(\mathcal{M}_{\Sigprime})$, and let  $\peq^{{\rm cut}}$ be the accessibility relation in $\cut$.

\begin{lemma}\label{lem:diamGrz}
Let $\ps \varphi\in\Sigma$ and suppose for some $T\in W_c$ we have $\ps \varphi\in T$, $\varphi \not \in T$. Then there is $S\in W_c$ such that $T_{\Sigprime} \prec_{\Sigprime} S_{\Sigprime}$ and $\mathcal{M}_{\Sigprime}, S_{\Sigprime} \models \varphi$.
\end{lemma}

\begin{proof}
From $\ps \varphi \in T$, by axiom $\sf Grz$, we obtain $\ps (\varphi \land \lnot \ps (\lnot \varphi \land \ps \varphi))\in T$. Then by Lemma~\ref{l:diam-witnessing} there is $S\seq_c T$ with $\varphi \land \lnot \ps (\lnot \varphi \land \ps \varphi)\in S$. By the construction of $\Sigprime$, we have $\lnot \ps (\lnot \varphi \land \ps \varphi)\in\Sigprime$. Hence $T_{\Sigprime}\peq_{\Sigprime} S_{\Sigprime}$, $\mathcal{M}_{\Sigprime}, S_{\Sigprime} \models \varphi$ and $\mathcal{M}_{\Sigprime}, S_{\Sigprime} \models \lnot \ps (\lnot \varphi \land \ps \varphi)$. To see that in fact $T_{\Sigprime}\prec_{\Sigprime} S_{\Sigprime}$, suppose the contrary, that $S_{\Sigprime} \peq_{\Sigprime} T_{\Sigprime}$, and note that $\lnot \ps (\lnot \varphi \land \ps \varphi)$ can be rewritten as $\nc (\ps \varphi \to \varphi)$. Then since $\mathcal{M}_{\Sigprime}, T_{\Sigprime} \models \ps \varphi$, we get $\mathcal{M}_{\Sigprime}, T_{\Sigprime} \models \varphi$. But this implies $\varphi \in T$, which contradicts $\varphi \not \in T$.
\end{proof}

\begin{lemma}\label{lem:path}
For any $\reach \varphi \psi \in \Sigma$, if $\mathcal{M}_{\Sigprime}, T_{\Sigprime} \models \reach \varphi \psi$, then there exists an up-down path $(v_0,v_1,\ldots,v_k)$ witnessing this such that for all $i<k$, either $v_i=v_{i+1}$, $v_{i} \prec_{\Sigprime} v_{i+1}$, or  $v_{i} \succ_{\Sigprime} v_{i+1}$.
\end{lemma}

\begin{proof}
If $\mathcal{M}_{\Sigprime}, T_{\Sigprime} \models \reach \varphi \psi$ then by Lemma~\ref{lemmUDR} we have a path $(v_0,v_1,\ldots,v_k)$ satisfying all of the above properties except possibly when $i=0$ or $i=k-1$.
If $v_0 \prec_{\Sigprime} v_1$ fails and $v_0 \neq v_1$ then $v_0 \peq_{\Sigprime} v_1$ and $v_0 \seq_{\Sigprime} v_1$.
By lemmas~\ref{lem:diamU} and \ref{lem:diamGrz} there is $S\in W_c$ such that $\mathcal{M}_{\Sigprime}, S_{\Sigprime} \models   \varphi  $ and $T_{\Sigprime} \prec_{\Sigprime} S_{\Sigprime}$.
We thus obtain a new up-down path $(\tilde  v_0,\ldots,\tilde  v_k) : = (v_0,S_{\Sigprime},v_2,\ldots,v_k)$ where $\tilde  v_0 \prec_{\Sigprime} \tilde  v_1$.
If $\tilde  v_{k-1} \succ_{\Sigprime} \tilde  v_k$ fails and $\tilde  v_{k-1} \neq \tilde  v_k$, we likewise choose $U $ such that $\mathcal{M}_{\Sigprime}, U_{\Sigprime} \models   \varphi  $ and $\tilde  v_k \prec_{\Sigprime} U_{\Sigprime}$, and replace $\tilde  v_{k-1}$ by $U_{\Sigprime}$.
The resulting path has all required properties.
\end{proof}

\smallskip

\begin{lemma}\label{lem:cutmodel}
If $\varphi \in \Sigma$ and $T\in W_c$ is arbitrary, we have:
\[\mathcal{M}_{\Sigprime}, T_{\Sigprime} \models \varphi \iff \cut, T_{\Sigprime} \models \varphi.\]
\end{lemma}

\begin{proof}
Proceed by induction on $\varphi$.
Atomic proposition and Boolean cases are immediate and thus omitted.

Consider the formula $\ps \varphi$.

We first show the right-to-left implication. Suppose $\cut, T_{\Sigprime} \models \ps \varphi$. Then there is $S_{\Sigprime} \seq^{{\rm cut}} T_{\Sigprime}$ such that $\cut, S_{\Sigprime} \models \varphi$. By the definition of $\peq^{{\rm cut}}$ we obtain $T_{\Sigprime} \peq S_{\Sigprime}$ and by the inductive hypothesis we have $\mathcal{M}_{\Sigprime}, S_{\Sigprime} \models \varphi$. Then $\mathcal{M}_{\Sigprime}, T_{\Sigprime} \models \ps \varphi$.

Now for the left-to-right implication. We have that $\mathcal{M}_{\Sigprime}, T_{\Sigprime} \models \ps \varphi$. Then we have two cases:
\begin{itemize}
	\item $\mathcal{M}_{\Sigprime}, T_{\Sigprime} \models \varphi$: 
	by the IH, $\cut, T_{\Sigprime} \models 
	\varphi$ and hence $\cut, T_{\Sigprime} \models \ps \varphi$.
	\item $\mathcal{M}_{\Sigprime}, T_{\Sigprime} \not \models \varphi$: 
	 by Lemma~\ref{lem:diamGrz}
	there is $S_{\Sigprime} \succ_{\Sigprime} T_{\Sigprime}$ with
	$\mathcal{M}_{\Sigprime}, S_{\Sigprime} \models \varphi$. 
	Again by the IH, $\cut, S_{\Sigprime} \models
	\varphi$ and it follows that $\cut, T_{\Sigprime} \models \ps \varphi$. 
\end{itemize}

Consider now the case of a formula $\reach \varphi\psi$.

The right-to-left implication is trivial, as $T_{\Sigprime} \peq^{{\rm cut}} S_{\Sigprime} \implies T_{\Sigprime} \peq_{\Sigprime} S_{\Sigprime}$, and hence any up-down path in $\cut$ witnessing $\reach \varphi\psi$ is also an up-down path in $\mathcal{M}_{\Sigprime}$ witnessing, by IH,  $\reach \varphi\psi$.

As for the left-to-right implication, if $\mathcal{M}_{\Sigprime}, T_{\Sigprime} \models \reach\varphi\psi$, by Lemma~\ref{lem:path} there exists an up-down path $(v_0,v_1,\ldots,v_{k})$ in $\mathcal{M}_{\Sigprime}$ witnessing $\reach \varphi \psi$ such that for all $i<k$, either $v_i = v_{i+1}$, $v_i\prec_{\Sigprime} v_{i+1}$, or  $v_i\succ_{\Sigprime} v_{i+1}$.
It is  immediate from the latter that this is also an up-down path in $\cut$. Using the IH, we have $\cut, T_{\Sigprime} \models \reach \varphi \psi$.
\end{proof}

\medskip

With this, we obtain completeness for finite poset models.

\begin{proposition}\label{propGrz}
If $\varphi$ is valid over the class of all finite poset models, then ${\sf PLR}\vdash \varphi$.
\end{proposition}

\begin{proof}
Arguing by contraposition, if $\neg\varphi$ is consistent with $\sf PLR$ then there is a theory $T$ with $\neg\varphi\in T$. Let $\Sigma$ be the smallest adequate set containing $\neg\varphi$, which is finite by Lemma~\ref{l:adequate-finite}. Then Lemma~\ref{lem:quotient} together with Lemma~\ref{lem:cutmodel} yield $\cut, T_{\Sigprime} \models \neg\varphi$, in other words $\cut, T_{\Sigprime} \not\models \varphi$. Since $\cut$ is a finite poset, the proof is finished.
\end{proof}

\medskip

Putting together Proposition~\ref{propGrz}, Proposition~\ref{prop:sound} and Theorem~\ref{theoKripkePoly}, we obtain our main result.

\begin{theorem}
Given a formula $\varphi$, the following are equivalent:

\begin{enumerate}

\item ${\sf PLR}\vdash\varphi$;

\item $\varphi$ is valid over the class of all finite poset models;

\item $\varphi$ is valid over the class of all polyhedral models.

\end{enumerate}
\end{theorem}

\section{Conclusions}\label{sec:conc}

Polyhedral semantics of modal logics with reachability operators enables reasoning about many interesting real-world scenarios. In this work, we provided sound and complete axiomatization of the polyhedral reachability logic. 

\vspace{2mm}

There are many different directions for future research from both theoretical and practical perspectives. Here we focus on some theoretical questions. Having obtained axiomatization and polyhedral completeness for modal logic with reachability, as the next step one could study more general topological completeness for such logics. On general topological models some of the theorems of ${\sf PLR}$ are no longer valid. Namely, $\ps$ is no longer definable through $\gamma$. This makes the axiomatization more challenging. Natural  follow up problems are axiomatizing the reachability logics for the class of all topological spaces and for other interesting classes of spaces such as Euclidean spaces, hereditarily irresolvable spaces, scattered spaces, locally connected spaces, etc.

The proposed axiomatisation may also have an impact on applications. Indeed, in light of the results obtained in~\cite{CGLMV23,CGGLMV23,BCGJLMV24}, the proposed axioms can be used both for automated reasoning and for model checking, in order to perform minimisation of formulas: a way to do this is to mix the existing model checking approach with equality saturation tools (see e.g.~\cite{WNWYFTP21}).

\medskip

\textbf{Acknowledgements} We thank Mamuka Jibladze and Evgeny Kuznetsov for fruitful discussions about the polyhedral semantics of reachability modality. The research was supported by the bilateral project between CNR (Italy) and SRNSF (Georgia) “Model Checking for Polyhedral Logics” (grant \#CNR\nobreakdash\nobreakdash-22\nobreakdash-010), by the Shota Rustaveli National Science Foundation of Georgia grant \#FR\nobreakdash-22\nobreakdash-6700, by European Union - Next GenerationEU - National Recovery and Resilience
Plan (NRRP), Investment 1.5 Ecosystems of Innovation, Project “Tuscany
Health Ecosystem” (THE), CUP: B83C22003930001 and by 
European Union - Next-GenerationEU - National Recovery and Resilience
Plan (NRRP) – Mission 4 Component 2, Investment N. 1.1, CALL PRIN 2022
D.D. 104 02-02-2022 – (Stendhal) CUP N. B53D23012850006.

\bibliographystyle{aiml}
\bibliography{reachability.bib}

\begin{thebibliography}{10}
\expandafter\ifx\csname url\endcsname\relax
  \def\url#1{\texttt{#1}}\fi
\expandafter\ifx\csname urlprefix\endcsname\relax\def\urlprefix{URL }\fi
\newcommand{\enquote}[1]{``#1''}

\bibitem{DBGM22}
Adam-Day, S., N.~Bezhanishvili, D.~Gabelaia and V.~Marra, \emph{Polyhedral
  completeness of intermediate logics: the nerve criterion}, The Journal of
  Symbolic Logic  (2022), pp.~1--41.

\bibitem{DBGM23}
Adam-Day, S., N.~Bezhanishvili, D.~Gabelaia and V.~Marra, \emph{The
  intermediate logic of convex polyhedra} (2023), arXiv:2307.16600 [math.LO].

\bibitem{HBSL}
Aiello, M., I.~Pratt{-}Hartmann and J.~{Benthem, van}, editors,
  \enquote{Handbook of Spatial Logics,} Springer, 2007.

\bibitem{BCLM19}
Belmonte, G., V.~Ciancia, D.~Latella and M.~Massink, \emph{Voxlogica: {A}
  spatial model checker for declarative image analysis}, in: \emph{Tools and
  Algorithms for the Construction and Analysis of Systems, {TACAS}},  LNCS
  \textbf{11427} (2019), pp. 281--298.

\bibitem{vBB07}
Benthem, J.~v. and G.~Bezhanishvili, \emph{Modal logics of space}, in:
  M.~Aiello, I.~Pratt{-}Hartmann and J.~v. Benthem, editors, \emph{Handbook of
  Spatial Logics}, Springer, 2007 pp. 217--298.

\bibitem{Bezh-Gab-2011}
Bezhanishvili, G. and D.~Gabelaia, \emph{Connected modal logics}, Archive for
  Mathematical Logic \textbf{50} (2011), pp.~287--317.

\bibitem{BCGGLM22}
Bezhanishvili, N., V.~Ciancia, D.~Gabelaia, G.~Grilletti, D.~Latella and
  M.~Massink, \emph{Geometric model checking of continuous space}, Logical
  Methods in Computer Science \textbf{Volume 18, Issue 4} (2022).

\bibitem{BCGJLMV24}
Bezhanishvili, N., V.~Ciancia, D.~Gabelaia, M.~Jibladze, D.~Latella, M.~Massink
  and E.~P. de~Vink, \emph{Weak simplicial bisimilarity for polyhedral models
  and slcs-eta}, in: V.~Castiglioni and A.~Francalanza, editors, \emph{Formal
  Techniques for Distributed Objects, Components, and Systems} (2024), pp.
  20--38.

\bibitem{BMMP18}
Bezhanishvili, N., V.~Marra, D.~McNeill and A.~Pedrini, \emph{Tarski's theorem
  on intuitionistic logic, for polyhedra}, Annals of Pure and Applied Logic
  \textbf{169} (2018), p.~ii.

\bibitem{BCGLM22}
Bussi, L., V.~Ciancia, F.~Gadducci, D.~Latella and M.~Massink, \emph{Towards
  model checking video streams using {V}ox{L}ogic{A} on {GPU}s}, in: J.~Bowles,
  G.~Broccia and R.~Pellungrini, editors, \emph{From Data to Models and Back},
  LNCS  \textbf{13268} (2022), pp. 78--90.

\bibitem{CZ97}
Chagrov, A. and M.~Zakharyaschev, \enquote{Modal logic,} Number~35 in Oxford
  Logic Guides, The Clarendon Press, Oxford University Press, New York, 1997.

\bibitem{CBLM21}
Ciancia, V., G.~Belmonte, D.~Latella and M.~Massink, \emph{A hands-on
  introduction to spatial model checking using voxlogica}, in: \emph{Model
  Checking Software} (2021), pp. 22--41.

\bibitem{CGLMV23}
Ciancia, V., D.~Gabelaia, D.~Latella, M.~Massink and E.~P. de~Vink, \emph{On
  bisimilarity for polyhedral models and slcs}, in: M.~Huisman and A.~Ravara,
  editors, \emph{Formal Techniques for Distributed Objects, Components, and
  Systems} (2023), pp. 132--151.

\bibitem{CGGLMV23}
Ciancia, V., J.~F. Groote, D.~Latella, M.~Massink and E.~P. de~Vink,
  \emph{Minimisation of spatial models using branching bisimilarity}, in:
  M.~Chechik, J.-P. Katoen and M.~Leucker, editors, \emph{Formal Methods}
  (2023), pp. 263--281.

\bibitem{CLLM16}
Ciancia, V., D.~Latella, M.~Loreti and M.~Massink, \emph{{Model Checking
  Spatial Logics for Closure Spaces}}, {Logical Methods in Computer Science}
  \textbf{Volume 12, Issue 4} (2016).

\bibitem{Engelking}
Engelking, R., \enquote{General Topology,} Sigma series in pure mathematics,
  Heldermann, 1989.

\bibitem{gabelaia2018}
Gabelaia, D., K.~Gogoladze, M.~Jibladze, E.~Kuznetsov and M.~Marx, \emph{Modal
  logic of planar polygons} (2018), arXiv:1807.02868 [math.LO].

\bibitem{gabelaiatacl}
Gabelaia, D., M.~Jibladze, E.~Kuznetsov and L.~Uridia, \emph{Characterization
  of flat polygonal logics} (2019), abstract of the talk given at the
  conference \emph{Topology, Algebra, and Categories in Logic}, Nice.
  https://math.unice.fr/tacl/assets/2019/abstracts.pdf.

\bibitem{Tarski44}
McKinsey, J. C.~C. and A.~Tarski, \emph{The algebra of topology}, Annals of
  Mathematics \textbf{45} (1944), pp.~141--191.

\bibitem{WNWYFTP21}
Willsey, M., C.~Nandi, Y.~R. Wang, O.~Flatt, Z.~Tatlock and P.~Panchekha,
  \emph{egg: Fast and extensible equality saturation} \textbf{5} (2021).
\newline\urlprefix\url{https://doi.org/10.1145/3434304}

\end{thebibliography}

\end{document}